\documentclass[a4paper]{article}
\usepackage[margin=3cm]{geometry}
\usepackage{latexsym}
\usepackage{amsmath}
\usepackage{amsthm}
\usepackage{amssymb}
\usepackage{amsfonts}
\usepackage[pdftex]{graphicx}

\newtheorem{theorem}{Theorem}
\newtheorem{remark}{Remark}
\newtheorem{definition}[theorem]{Definition}

\newtheorem{notation}[theorem]{Notation}

\newtheorem{lemma}{Lemma}

\begin{document}

\title{Synchronization of piece-wise continuous systems of fractional order}

\author{MARIUS-F. DANCA\\Department of Mathematics and Computer Science, Avram Iancu University, \\Str. Ilie Macelaru, nr. 1A, 400380 Cluj-Napoca, Romania,\\and\\Romanian Institute of Science and Technology, \\Str. Ciresilor nr. 29, 400487 Cluj-Napoca, Romania}

\maketitle

\begin{abstract}

The aim of this study is to prove analytically that synchronization of a piece-wise continuous class of systems of fractional order can be achieved. Based on our knowledge, there are no numerical methods to integrate differential equations with discontinuous right hand side of fractional order which model these systems. Therefore, via Filippov's regularization \cite{fil} and Cellina's Theorem \cite{aubin,aubin2}, we prove that the initial value problem can be converted into a continuous problem of fractional-order, to which numerical methods for fractional orders apply. In this way, the synchronization problem transforms into a standard problem for continuous systems of fractional order. Three examples of fractional-order piece-wise systems are considered: Sprott system, Chen and Shimizu-–Morioka system.
\end{abstract}


\section{Introduction}

Discontinuous fractional-order systems provide a logical link between the fractional derivative approach to descriptive system properties, such as "memory" and "heredity", and the physical system properties, such as dry friction, forced vibration brake processes with locking phase, stick, and slip phenomena.

However, on our knowledge, there are very few works (if any), on discontinuous systems of fractional-order and the known synchronization algorithms apply to continuous systems of integer or fractional-order and rarely to discontinuous systems of integer order.

Also, even most of dedicated numerical methods for DE of fractional-order can be used to ``integrate'' abrupto discontinuous equations of fractional-order, this approach has not any mathematically justification (it is known that discontinuous equations may have not any solutions). Therefore, special numerical methods and approach are necessary in this case.

Nowadays, there are numerical methods for continuous DE of fractional-order (see e.g. \cite{kai1,kai2}) and also for DE of integer order with discontinuous right hand side (see e.g. \cite{dont,kas}).

Therefore, modeling continuously discontinuous systems of fractional-order, could be of a real interest in synchronization, chaos control, anticontrol but also for quantitative analysis.

The fractional-order systems considered in this paper are modeled with piece-wise continuous functions $f:\mathbb{R}^n\rightarrow \mathbb{R}^n$, of the following form
\begin{equation}\label{f}
f(x(t))=g(x(t))+Kx(t)+A(x(t))s(x(t)),
\end{equation}

\noindent where $g:\mathbb{R}^n\rightarrow \mathbb{R}^n$ is a vector single-valued, nonlinear and at least continuous function, $s:\mathbb{R}^n\rightarrow \mathbb{R}^n$, $s(x)=(s_1(x_1),s_2(x_2),...,s_n(x_n))^T$ a vector valued piece-wise function, with $s_i:\mathbb{R}\rightarrow \mathbb{R}$, $i=1,2,...,n$ real piece-wise constant functions, $A\in \mathbb{R}^{n\times n}$ a square matrix of real functions and $K\in \mathbb{R}^{n\times n}$ a square constant real matrix, $Kx$ representing the linear part of $f$.

\begin{notation}
Let denote by $\mathcal{M}$ the discontinuity set of $f$ \textup{(}of zero Lebesgue measure: $\mu(\mathcal{M})=0$\footnote{As known, the Legesgue measure of a point on the real line, as well the Lebesgue measure of a line in $\mathbb{R}^2$, or Lebesgue measure of a plane in $\mathbb{R}^3$, is zero.}\textup{)}, generated by the discontinuity points of the components $s_i$.
\end{notation}

\noindent $\mathcal{M}$ separates $\mathbb{R}^n$ in several sub-domains $\mathcal{D}_i$, where $f$ is continuous, and possible differentiable in their interior.

\noindent The following assumption will be considered

\vspace{3mm}
\noindent (\textbf{H1}) $As$ is discontinuous in at least one of his components.
\vspace{3mm}

For example, the following piece-wise continuous (linear) function $f:\mathbb{R}\rightarrow \mathbb{R}$
\begin{equation}\label{exemplu}
f(x)=2-3sgn(x),
\end{equation}

\noindent has $\mathcal{M}=\{0\}$ which determines on $\mathbb{R}$ the continuity sub-domains $\mathcal{D}_1=(-\infty, 0]$ and $\mathcal{D}_2=[0,\infty)$. The graph is plotted in Fig. \ref{fig1}a.

The form of $f$, given by (\ref{f}), appears in the great majority of nonlinear piece-wise continuous systems of fractional or integer order, which are modeled by the following Initial Value Problem (IVP)
\begin{equation}\label{IVP0}
D_*^q x(t)=f(x(t)):=g(x(t))+Kx(t)+A(x(t))s(x(t)),~~~x(0)=x_0,~~~ t\in I=[0,\infty).
\end{equation}
\noindent In this paper, $D_*^q$, with $q=(q_1,q_2,...,q_n)$, $0<q_i\leq1$, $i=1,2,...,n$ ($q=1$ for the integer order), denotes the commonly used operator in fractional calculus: Caputo's differential operator of order $q$ (called also \emph{smooth fractional derivative} with starting point 0)\cite{old,cap,pod}
\begin{equation*}
D_*^qx(t)=\frac{1}{\Gamma(1-q)}\int_0^t (t-\tau)^{-q}\frac {d}{dt}x(\tau)d\tau,
\end{equation*}

\noindent where $\Gamma$ is the Euler's Gamma function
\begin{equation*}
\Gamma(z)=\int_0^t t^{z-1}e^{-t}dt, ~~~ z\in\mathbb{C}, ~~Re(z)>0.
\end{equation*}

We consider in this paper $\mathbb{R}^3$ examples. For example, the fractional variant of the piece-wise Chen system \cite{aziz}
\begin{equation}\label{chen}
\begin{array}{l}
D_{\ast }^{q_{1}}x_{1}=a\left( x_{2}-x_{1}\right) , \\
D_{\ast }^{q_2} x_{2}=\left( c-a-x_{3}\right) sgn(x_{1})+cdx_{2}, \\
D_{\ast }^{q_3} x_{3}=x_{1}sgn(x_{2})-bx_{3}.%
\end{array}%
\end{equation}

\noindent with $a=1.18$, $b=0.16$, $c=1.2$, $d=0.1$ and the fractional-order $(q_1,q_2,q_3)$, has $g(x)=\left( 0,0,0\right) ^{T}$ (i.e. $f$ in this case is piece-wise linear), $s(x)=(sgn(x_{1}),sgn(x_{2}),\allowbreak sgn(x_{3}))^{T}$ and
\[
K=\left(
\begin{array}{ccc}
-a & a & 0 \\
0 & cd & 0 \\
0 & 0 & -b%
\end{array}%
\right),
\]
\noindent and
\[
A(x)=\left(
\begin{array}{ccc}
0 & 0 & 0 \\
c-a-x_{3} & 0 & 0 \\
0 & x_{1} & 0%
\end{array}%
\right).
\]

As in most practical examples, the use of Caputo derivative in the IVP (\ref{IVP0}) is fully justified since in these problems we need physically interpretable initial conditions, or Caputo derivative satisfies these demands, by avoiding the expression of initial conditions with fractional derivatives \cite{kai1}. Accordingly, the initial condition in (\ref{IVP0}), can be considered in the standard form $x(0) = x_0$.

To overcome the discontinuity impediment, we shall use Filippov's approach \cite{fil}. This technique targets the piece-wise constant functions $s$, and converts them in set-valued functions. Next, via Cellina's Theorem \cite{aubin,aubin2}, the set-valued functions are continuously approximated in small neighborhoods of underlying set-valued functions.

The results are valid for a large class such as \emph{Heaviside} function $H$, \emph{rectangular} function (as difference of two Heaviside functions), or \emph{signum}, one of the most encountered functions in practical applications.

The paper is organize as follows: Section \ref{switch} deals with the approximation of $f$ defined by (\ref{f}) and shows how the IVP (\ref{IVP0}) can be transformed into a continuous single valued problem. In Section 3 the asymptotically synchronization of piece-wise continuous systems of fractional-order is investigated and the necessary condition for chaotic behavior of these system is presented. In Section 4 the asymptotically synchronization is applied to three piece-wise continuous systems of fractional-order: Chen's system, Sprott's system and Shimizu-–Morioka's system. Appendix includes proofs and results utilized in the paper.

\section{Continuous approximation of $f$ }\label{switch}

In this section we prove that the considered class of piece-wise continuous functions defined in (\ref{f}), can be approximated as closely as desired with continuous functions. First, the piece-wise continuous function $f$ will be transformed into a set-valued function, which will be approximated with continuous functions. For this purpose, we will choose the way proposed by Fillipov in \cite{fil}, namely the \emph{Filippov regularization}. Thus, the discontinuous function $f$ is transformed into a convex set-valued function $F$ into the set of all subsets of $\mathbb{R}^n$, $F:\mathbb{R}^n\rightrightarrows \mathbb{R}^n$. One of the simplest expressions for $F$, is \cite{fil,aubin,aubin2}
\begin{equation}\label{fill}
F(x)=\bigcap_{\varepsilon >0}\bigcap_{\mu(\mathcal{M})=0} \overline{conv}(f({z\in \mathbb{R}^n: |z-x|\leq\varepsilon}\backslash \mathcal{M})).
\end{equation}
$F(x)$ is the convex hull of $f(x)$, $\mu$ being the Lebesgue measure and $\varepsilon$ the radius of the ball centered in $x$. At the points where $f$ is continuous, $F(x)$ consists of one single point, which coincides with the value of $f$ at this point (i.e. we get back $f(x)$ as the right hand side: $F(x)=\{f(x)\}$). In the points belonging to $\mathcal{M}$, $F(x)$ is given by (\ref{fill}).

If the piece-wise-constant functions $s_i$ are $sgn$ functions, their set-valued form, obtained with Filippov regularization, denoted by $Sgn:\mathbb{R}\rightrightarrows \mathbb{R}$, is defined as follows (see Fig. \ref{fig2}a) before regularization and Fig. \ref{fig2}b) after regularization)
\begin{equation}
Sgn(x)=\left\{
\begin{array}{cc}
\{-1\}, & x<0, \\
\lbrack -1,1], & x=0, \\
\{+1\}, & x>0.%
\end{array}%
\right.
\end{equation}

\noindent By applying the Filippov regularization to $f$, one obtains the following set-valued function
\begin{equation}
\label{IVP1}
F(x)=g(x)+Kx+A(x)S(x),
\end{equation}

\noindent with
\begin{equation}\label{s}
S(x)=(S_1(x_1),S_2(x_2),...,S_n(x_n))^T,
\end{equation}

\noindent $S_i:\mathbb{R}\rightarrow \mathbb{R}$ being the set-valued variants of $s_i$, $i=1,2,...,n$ ($Sgn(x_i)$ in the usual case of $sgn(x_i)$).

\noindent For example, the graph of the set-valued variant of $f$ defined in (\ref{exemplu}) is plotted in Fig. \ref{fig1}b.

The notions and results presented next are considered in $\mathbb{R}$, but they are also valid in the general case $\mathbb{R}^n$, $n>1$. Let a set-valued function $F:\mathbb{R}\rightrightarrows \mathbb{R}$.

\noindent A set-valued function $F$ can be characterized by its graph
\begin{equation*}
Graph(F):=\{(x,y)\in \mathbb{R}\times \mathbb{R}, ~y\in F(x)\}.
\end{equation*}

\begin{remark}\label{rr}
Due to the symmetric interpretation of a set-valued function as a graph (see e.g. \cite{aubin}) we shall say that a set-valued function satisfies a property if and only if its graph satisfies it. For instance, a set-valued function is said to be closed if and only if its graph is closed.
\end{remark}

\begin{definition}
A set-valued function $F$ is upper semicontinous (u.s.c.) at $x^0\in \mathbb{R}$, if for any open set $E$ containing $F(x^0)$, there exists a neighborhood $A$ of $x^0$ such that $F(A)\in B$.
\end{definition}

\noindent We say that $F$ is u.s.c. if it is so at every $x^0\in \mathbb{R}$.

\noindent U.s.c., which is a basic property, practically means that the graph of $F$ is closed.

\begin{definition}
A single-valued function $h:\mathbb{R}\rightarrow \mathbb{R}$ is called an \emph{approximation} (\emph{selection}) of the set-valued function $F$ if
\begin{equation*}
\forall	x\in \mathbb{R},~~h(x)\in F(x).
\end{equation*}
\end{definition}

\noindent Generally, a set-valued function admits (infinitely) many approximations (see Fig. \ref{fig1}b for the case of function defined in (\ref{exemplu})).

As proved in \cite{dan1}, the set-valued functions $S_i$, $i=1,2,...,n$, can be approximated due to the Approximate Theorem, called also Cellina's Theorem (Appendix) which states that a set-valued function $F$, with closed graph and convex values, admits continuous approximations. This result is assured by the following lemma

\begin{lemma}\label{tprinc}
For every $\varepsilon>0$, the set-valued functions $S_i$, $i=1,2,...,n$ admit continuous approximations in the $\varepsilon$-neighborhood of $S_i$.
\end{lemma}
\begin{proof} $S_i$, for $i=1,2,...,n$, are convex u.s.c. (see e.g. the Remark in \cite{fil} p. 43 or the Example in \cite{aubin2} p. 39 for u.s.c.) and, via Remark \ref{rr}, are non-empty closed valued functions. Therefore, they verifies Cellina's Theorem which guaranties the existence of continuous approximations on $\mathbb{R}$.
\end{proof}

\begin{notation}
Let denote by $\widetilde{s}_i:\mathbb{R}\rightarrow \mathbb{R}$ the approximations of $S_i$.
\end{notation}

\noindent For the sake of simplicity, for each component $\widetilde{s}_i(x_i)$, $i=1,2,...,n$, $\varepsilon_i$ will be considered as having the same value.

\noindent Some of the best candidates for $\widetilde{s}$ are the \emph{sigmoid} functions, since they provide the required flexibility and to which the abruptness of the discontinuity can be easily modified. For $S(x)=Sgn(x)$, one of the most utilized sigmoid approximations is the following function $\widetilde{sgn}$\footnote{The class of sigmoid functions includes for example the ordinary arctangent such as $\frac{2}{\pi}arctan\frac{x}{\varepsilon}$, the hyperbolic tangent, the error function, the logistic function, algebraic functions like $\frac{x}{\sqrt{\epsilon+x^2}}$, and so on.}
\begin{equation}\label{h_simplu}
\widetilde{sgn}(x)=\frac{2}{1+e^{-\frac{x}{\delta}}}-1\approx Sgn(x),
\end{equation}

\noindent where $\delta$ is a positive parameter which controls the slope in the neighborhood of the discontinuity $x=0$ (Fig. \ref{fig3}a and Fig. \ref{fig3}b).

Summarizing, we can enounce the following result, which assures the possibility to approximate continuously $f$
\begin{theorem}\label{th1}
Let $f$ defined by (\ref{f}). If $g$ is continuous, then there exist continuous approximations of $f$, $\tilde{f}:\mathbb{R}^n\rightarrow \mathbb{R}^n$
\begin{equation}\label{glo}
\tilde{f}(x)=g(x)+Kx+A(x)\widetilde{s}(x)\approx f(x).
\end{equation}

\end{theorem}

\noindent Specifically, Theorem \ref{th1} actually means that the considered function $f$ can be approximated simply by replacing $s$ with $\tilde{s}$ (dotted line in sketch in Fig. \ref{schema}). For example, $f$ defined by (\ref{exemplu}), can be approximated as follows
\begin{equation}
\tilde{f}(x)=2-3\widetilde{sgn}(x)=2-3\left(\frac{2}{1+e^{-\frac{x  }{\delta}}}-1\right).
\end{equation}

\noindent Theorem \ref{th1} states that systems modeled by the IVP (\ref{IVP0}), can be continuously approximated by the following continuous IVP
\begin{equation*}
D_*^q(x)=\tilde{f}(x),~~~x(0)=x_0,
\end{equation*}

\noindent with $\tilde{f}$ defined by (\ref{glo}).

\section{Synchronization}\label{sincro}
Once we proved that systems modeled by (\ref{IVP0}) can be continuously approximated, they can be synchronized via any kind of synchronization schemes for continuous systems. In this paper we consider the synchronization in coupled chaotic system via \emph{master-slave} configuration.

As known, a linear autonomous system of fractional-order is \emph{asymptotically stable }if his zero (equilibrium) point is asymptotically stable.

The computation of the Jacobian requires the following assumption
\vspace{3mm}

\noindent (\textbf{H2})
Function $g$ in (\ref{f}) is suppose to be differentiable on $\mathbb{R}^n$.
\vspace{3mm}

Since the discontinuous functions appearing in the considered examples are $sgn$, next we study some properties of this function and its approximations.

Let $X^*$ and $\tilde{X}^*$ the equilibrium points of $f$ and $\tilde{f}$ respectively, and $J$ and $\tilde{J}$ the related Jacobians.

\vspace{3mm}
\textbf{Property 1}\label{prop1}
For every $\delta>0$, there exists a small neighborhood of $\mathcal{M}$, $\mathcal{V}$, depending on $\delta$, such that $\tilde{X}^*\approx X^*$ for $x\not\in\mathcal{V}$.
\vspace{3mm}

\noindent See the proof in Appendix.

\begin{remark}\label{cater}
As known, the error of ABM method utilized in this paper to integrate the fractional DE, is of order $O(h^p)$ \cite{kai1} with $p = min(2,1+q_{min})$. For our step size $h=0.005$, this error is of order of $1e-5$. On the other side, in order to ensure the validity of Property \ref{prop1} for a large class of systems (\ref{IVP0}), the size of $\mathcal{V}$ must be smaller than $h^p$. $\delta=1/100000$, proves to be an acceptable compromise between the numerical accuracy and computer precision and also assures the requirements for Property \ref{prop1}. For this choice, $\mathcal{V}=(-1.589e-4,1.589e-4)$ and for $x\not\in\mathcal{V}$, the difference between $\widetilde{sgn}$ and the branch $\pm1$ of the function $sgn$ is of order of $1e-7$. In all studied examples, $X^* (\tilde{X}^*)$ are situated outside of these neighborhoods.
\end{remark}

\noindent Regarding the size of $\delta$ the following hypothesis is considered

\vspace{3mm}
\noindent (\textbf{H3})\label{del} In this paper we chosen $\delta=1/100000$.
\vspace{3mm}

\noindent For the derivative $\frac{d}{dx}\widetilde{sgn}$ (plotted as function on $\delta$ in Fig. \ref{fig55}a), the following property holds
\vspace{3mm}
\textbf{Property 2}\label{prop2}
Assume \textup{(\textbf{H2})}. For every $\delta>0$, there exists a neighborhood of $x=0$, $\mathcal{V}$, depending on $\delta$ such that $\tilde{J}|_{\widetilde{X}^*}\approx J|_{X^*}$, for $x\not\in\mathcal{V}$.
\vspace{3mm}

\noindent See the proof in Appendix.

Let consider the following piece-wise continuous master system of fractional-order
\begin{equation}\label{master0}
D^q_*x=f(x):=g(x)+Kx+A(x)s(x), ~~~x(0)=x_0,
\end{equation}

\noindent and the slave system
\begin{equation}\label{slave0}
D^q_*y=f(y)+u:=g(y)+Ky+A(y)s(y)+u,~~~y(0)=y_0,
\end{equation}

\noindent where $u \in \mathbb{R}^n$ is the control designed such as the state of the slave system (\ref{slave0}) evolves as the states of the master system (\ref{master0}).

\noindent After continuous approximation, the master system becomes
\begin{equation}\label{master}
D^q_*x=\tilde{f}(x):=g(x)+Kx+A(x)\tilde{s}(x),~~~x_0=x(0),
\end{equation}
\noindent and the slave system
\begin{equation}\label{slave}
D^q_*y=\tilde{f}(y)+u:=g(y)+Ky+A(y)\tilde{s}(y)+u,~~~y_0=y(0).
\end{equation}
The utilized active control method requires to design $u \in \mathbb{R}^n$ such as the error, defined as  $e=y-x$, tends asymptotically to zero: $\underset{t\rightarrow \infty }{\lim }||e(t)||=0$ ($||\cdot||$ being the Euclidean norm). Thus, the error dynamical system is obtained by subtracting (\ref{master}) from (\ref{slave})
\begin{equation}\label{error}
D^q_*e=g(y)-g(x)+Ke+A(y)\tilde{s}(y)-A(x)\tilde{s}(x)+u,~~~e(0)=y(0)-x(0).
\end{equation}
\noindent and the asymptotically synchronization transforms into asymptotically stability of the zero equilibrium point of (\ref{error}). For this purpose, $u$ has to be defined such as the error system becomes an asymptotically stable linear system. The usual choice for the active control is
\begin{equation*}
u=-g(y)+g(x)-A(y)\widetilde{s}(y)+A(x)\widetilde{s}(x)+v,
\end{equation*}
\noindent with $v\in \mathbb{R}^n$, $v=Me$, where $M$ is some real square matrix $M\in R^{n\times n}$, which can be chosen in many possible ways. Thus, by replacing $u$ in (\ref{error}), the error system becomes a linear system
\begin{equation}\label{error2}
D^q_*e=Ee, ~~~e(0)=y(0)-x(0),
\end{equation}

\noindent with $E=M+K$.

\begin{theorem}\label{stabt}
The piece-wise continuous master-slave system \textup{(\ref{master0})-(\ref{slave0})} asymptotically synchronizes if and only if:
\itemize
\item[a.] for the commensurate case, all eigenvalues $\lambda$ of $E$, verify the condition
\begin{equation}\label{stab1}
|arg(\lambda)|>q \pi/2;
\end{equation}\item[b.] for the incommensurate case $q_i=k_i/m_i<1$, $k_i$, $m_i\in \mathbb{N}$, $m_i\neq 0$, for $i=1,2,...,n$, $k_i$, $m_i$ being coprime positive integers, $(k_i,m_i)=1$, all the roots of the characteristic equation
\begin{equation}\label{car}
P(\lambda):=det(diag [\lambda^{mq_1},\lambda^{mq_2},...,\lambda^{mq_n}]-E)=0,
\end{equation}

\noindent with $m$ the least common multiple of the denominators $m_i$ verify the condition
\begin{equation}
\label{stab2}
|arg(\lambda)|> \pi/2m.
\end{equation}

\enditemize

\end{theorem}

\begin{proof}
Under the assumptions given by Properties \ref{prop1} and Property \ref{prop2}, the master and the slave systems transform into continuous systems of fractional-order and the proof follows the same steps such as the proof for original theorems for continuous systems (see \cite{maty,deng} for commensurate and incommensurate cases respectively).
\end{proof}

Summarizing, if we denote with $\Lambda$ the spectrum of the eigenvalues of $E$ or of the roots of (\ref{car}), and with $\alpha_{min}=min\{|arg(\Lambda)\}$, the sufficient and necessary asymptotically synchronization conditions (\ref{stab1}) and (\ref{stab2}) can be written as follows
\begin{equation}
\label{stab3}
\alpha_{min}> \gamma\pi/2,
\end{equation}

\noindent where $\gamma=q$ for the commensurate case, and $\gamma=1/m$ in the case of incommensurate case, or $\Lambda$ is included in the domain $\Omega$ defined as follows (Fig. \ref{fig66}a)
\[
\Lambda \subset\Omega =\{\lambda \in \mathbb{C},|\arg \{\lambda \}|>\gamma\pi/2\}.
\]

\begin{remark}

\itemize
\item[i)] Relations (\ref{stab1}), (\ref{stab2}) or  (\ref{stab3}) with $"\geq"$ instead $"<"$, mean that there are $\lambda_i$ situated on the separatrices $d_{1,2}$ having the equations $\pm tan(\gamma\pi/2)$ (Fig. \ref{fig66}). If those eigenvalues (or roots) situated on $d_{1,2}$ (for example the points $\lambda^*$ and $\bar{\lambda}^*$ in Fig. \ref{fig66}b which satisfy the equality), have geometric multiplicity of one, then the synchronization is only stable and not asymptotical stable.\footnote{The geometric multiplicity represents the dimension of the eigenspace of eigenvalues.}

\item[ii)] It is to note that $arg$ should be not considered $arctan$ function, since $arctan\in (-\pi/2,\pi/2)$, while $arg \in[-\pi,\pi]$. A possible choice is the function $atan2$ implemented in some software packages or, for example, the formulae  $arg(z)=arctan(y/x)+pi/2sign\\ (y)(1-sign(x))$, or $2arctan(\frac{\sqrt{x^2+y^2}-x}{y})$, where $z=x+iy$.
\enditemize
\end{remark}

\paragraph{Chaotic piece-wise continuous systems of fractional-order}

Since in this paper the synchronization deals with chaotic motions, we shall study computationally the existence of chaotic behaviors for the considered examples, beside a necessary criterion for chaos existence in fractional-order nonlinear systems, derived from the Stability Theorem \ref{stabt}.

\noindent We do not consider here the qualitative aspects of the equilibrium points but only the necessary condition for chaos and numerical evidences of chaotic motions (see e.g. \cite{tava2} for a study on the number of saddle points, underlying eigenvalues, one-scroll, double-scroll).

\noindent Hereafter, it is supposed that the determination of $\tilde{X}^*$ and $\tilde{J}$, are assured by Properties \ref{prop1} and \ref{prop2}.

\noindent The utilized numerical method is the Adamas-Bashforth-Moulton variant, proposed by Kai et al in \cite{kai1}, and the utilized step size $h=0.005$.

\noindent The trajectories corners which can be seen in phase plots and time series, are typical to discontinuous systems \cite{dont,kas}.

Let consider an approximated system of fractional-order (\ref{master}), with $\tilde{X}^*$ the equilibrium points and $\tilde{J}$ the Jacobian matrix.

\noindent The condition necessary for chaotic motion related to one of the equilibria $X^*$, is \cite{tava}
\begin{equation}\label{nece}
\alpha_{min}\leq\gamma\pi/2,
\end{equation}

\noindent or $\Lambda\subset\Phi=\{\lambda\in \mathbb{C}, |arg(\lambda)|\geq\alpha_{min}, ~\alpha_{min}\leq\gamma\pi/2\}$ (Fig. \ref{fig66}b).

In this case, $\alpha_{min}$ is determined for the Jacobian $\tilde{J}_{X^*}$, or for the roots of the characteristic equation
\begin{equation}\label{carhaos}
P(\lambda):=det(diag [\lambda^{mq_1},\lambda^{mq_2},...,\lambda^{mq_n}]-\tilde{J}_{X^*})=0
\end{equation}

\noindent Even ({\ref{nece}) is only a necessary condition for chaos, it is useful to find the minimum commensurate order $q$ for system (\ref{IVP0}) to remain chaotic. Thus, once we find $\Lambda$ and $\alpha_{min}$, by setting in (\ref{nece}) $\gamma=q$, the values of $q$ from which chaos might appear are
\begin{equation}\label{inec}
q>q_{min}=\frac{2}{\pi}\alpha_{min}.
\end{equation}

\noindent For example let us consider the Chen's system (\ref{chen})
\begin{equation}
\begin{array}{l}
D_{\ast }^{q_{1}}x_{1}=a\left( x_{2}-x_{1}\right) , \\
D_{\ast }^{q_2} x_{2}=\left( c-a-x_{3}\right)sgn(x_{1}) +cdx_{2}, \\
D_{\ast }^{q_3} x_{3}=x_{1}sgn(x_{2})-bx_{3}.%
\end{array}%
\end{equation}

\noindent and his the continuous approximation.
\begin{equation}\label{chen_apr}
\begin{array}{l}
D_{\ast }^{q_{1}}x_{1}=a\left( x_{2}-x_{1}\right) , \\
D_{\ast }^{q_2} x_{2}=\left( c-a-x_3\right) \widetilde{sgn}(x_{1})+cdx_{2}, \\
D_{\ast }^{q_3} x_{3}=x_{1}\widetilde{sgn}(x_{2})-bx_{3}.%
\end{array}%
\end{equation}

\noindent Beside the origin, the system has two other equilibria $\tilde{X}^*_{1,2}=(\pm 0.003, \pm0.003,0.020)$ and the Jacobian
\[
\tilde{J}=\left(
\begin{array}{ccc}
-a & a & 0 \\
0 & cd & -sgn(x_1) \\
sgn(x_2) & 0 & -b%
\end{array}%
\right) _{\tilde{X}^{\ast }_{1,2}}.
\]

\noindent The eigenvalues at $\tilde{X}_{1,2}^*$ are  $\Lambda=(-1.636,0.208+0.815i,\\0.208-0.815i)$, $arg\{\Lambda\}=\{3.141,1.321,-1.321\}$.

\noindent Let consider the case of $\tilde{X}^*_1=(0.003, 0.003,0.020)$, similar results being obtained for $\tilde{X}^*_2$.
\paragraph{Commensurate case}: $q=0.99$. Condition (\ref{nece}) is verified: $\alpha_{min}=1.321<1.555=0.99\pi/2$, $\Lambda\subset \Phi$ (Fig. \ref{fig6}a), and for this value of $q$, the system is chaotic (Fig. \ref{fig6}b).

\noindent The minim commensurate value of $q$ to have chaotic behavior is, via (\ref{inec}), $q>q_{min}=2\alpha_{min}/\pi =0.841$.

\paragraph{Incommensurate case}: $q=(1,0.9,1)$. The characteristic polynomial (\ref{carhaos}) is
\begin{equation*}
\lambda^{29}-3/25\lambda^{20}+67/50\lambda^{19}-201/1250\lambda^{10}+118/625\lambda^{9}+36167/31250.
\end{equation*}

\noindent Here, $m=10$, $\gamma=1/10$ and $\alpha_{min}=0.138<0.157=\pi/20$. Therefore condition (\ref{nece}) is verified and $\Lambda\subset \Phi$ (Fig. \ref{fig6}c). The system behaves chaotic (Fig. \ref{fig6}d).

\section{Applications}
In this section we synchronize three-dimensional systems: Sprott's systems, Chen's systems and Shimizu-–Morioka's systems, after which, we synchronize two different systems: Sprott's and Chen's systems.

For all numerical experiments, beside phase plots (where the trajectories of the master and slave systems are overplotted) and time series, we calculated the Hausdorff distance $d_H$ (Appendix). After few hundreds steps have been neglected, $d_H$ is of order of $1e-5$. Time series details (for $t\in[0,20]$) are also plotted to reveal the synchronization process. The roots of the characteristic equations have been calculated with Matlab function \emph{solve}.

The integration is made via ABM algorithm with the time step size $h=0.005$.
\subsection{Sprott system}

The Sprott's system \cite{sprotus,sprot2}, considered bellow as master system, has the following approximated form
\begin{equation}\label{spr_apr}
\begin{array}{l}
D_{\ast }^{q_1}x_{1}=x_{2}, \\
D_{\ast }^{q_2}x_{2}=x_{3}, \\
D_{\ast }^{q_3}x_{3}=-x_{1}-x_{2}-0.5x_{3}+\widetilde{sgn}(x_{1}),%
\end{array}%
\end{equation}

\noindent with the initial condition $x(0)=(0.29,0.12,0.22)^T$, and the slave system is
\begin{equation}\label{spr_slave}
\begin{array}{l}
D_{\ast }^{q_1}y_{1}=y_{2}+u_1, \\
D_{\ast }^{q_2}y_{2}=y_{3}+u_2, \\
D_{\ast }^{q_3}y_{3}=-y_{1}-y_{2}-0.5y_{3}+\widetilde{sgn}(y_{1})+u_3,
\end{array}%
\end{equation}

\noindent with the initial condition $y(0)=x(0)+(0.1,0.1,0.1)^T$, where the controller $u(t)=(u_1(t),u_2(t),u_3(t))^T\in \mathbb{R}^3$ has to be defined next. Subtracting (\ref{spr_apr}) from (\ref{spr_slave}), one obtains the error system
\begin{equation}\label{spr_er}
\begin{array}{l}
D_{\ast }^{q_1}e_{1}=e_{2}+u_1, \\
D_{\ast }^{q_2}e_{2}=e_{3}+u_2, \\
D_{\ast }^{q_3}e_{3}=-e_1-e_2-0.5e_3+\widetilde{sgn}(y_1)-\widetilde{sgn}(x_1)+u_3,
\end{array}%
\end{equation}

\noindent where $e_i=y_i-x_i$, $i=1,2,3$, are the synchronization errors with $e(0)=y(0)-x(0)$.

\noindent Here
\[
K=\left(
\begin{array}{ccc}
0 & 1 & 0 \\
0 & 0 & 1 \\
-1 & -1 & -0.5%
\end{array}%
\right),
\]

\noindent and by choosing the controller
\[
u=\left(
\begin{array}{c}
0 \\
0 \\
\widetilde{sgn}(x_{1})-\widetilde{sgn}(y_{1})%
\end{array}%
\right) ~+M\left(
\begin{array}{c}
e_{1} \\
e_{2} \\
e_{3}%
\end{array}%
\right), ~
\]

\noindent with
\[
M=\left(
\begin{array}{ccc}
0 & 0 & 0 \\
0 & 0 & 0 \\
0 & 0 & -1.5%
\end{array}%
\right),
\]

\noindent the error system receives the linear form
\begin{equation}
\begin{array}{l}
D_{\ast }^{q_1} e_{1}=e_2, \\
D_{\ast }^{q_2} e_{2}=e_3,\\
D_{\ast }^{q_3} e_{3}=-e_1-e_2-2e_3,
\end{array}
\end{equation}

\noindent with
\[
E=M+K=\left(
\begin{array}{ccc}
0 & 1 & 0 \\
0 & 0 & 1 \\
-1 & -1 & -2%
\end{array}%
\right).
\]

\paragraph{Commensurate case}: $q_1=q_2=q_3=0.92$. The eigenvalues of $E$ are $\lambda_1=-1.755$ and $\lambda_{2,3}= -0.123\pm0.745i$ with arguments $arg(\lambda_1)=\pi$ ($\lambda_1$ is located on the negative real axis where a complex number has $arg=\pi$), and $arg(\lambda_{2,3})= \pm1.734$. Therefore $\alpha_{min}=1.734>1.445=0.92\pi/2$ and all eigenvalues are located inside the asymptotically stability region $\Omega$ (Fig. \ref{fig7}a) and the systems synchronize (see Fig. \ref{fig7}b and Fig. \ref{fig7}c).

\begin{remark}
Generally, the controller $u$ for systems defined as $D^q_* x_i=x_{i+1}$, $i=1,2,...,n-1$, and $D_*^q x_n=f(x)$, can be defined for only the $n$th equation.
\end{remark}
\vspace{3mm}
\paragraph{Incommensurate case}: $q=(0.9,1,0.9)$. Let us consider the master-slave system (\ref{spr_apr})-(\ref{spr_slave}) with the same controller. $m=10$, and the characteristic polynomial (\ref{car}) is
\begin{equation}
P(\lambda):=\lambda^{28}+2\lambda^{19}+\lambda^9+1.
\end{equation}

\noindent The images of the $28$ complex roots are situated inside the stability region $\Omega$ and $\alpha_{min}=0.1814>0.1571=\pi/20$ (Fig. \ref{fig8}a) and the two systems synchronize (Fig. \ref{fig8}b and Fig. \ref{fig8}c).

\subsection{Chen system}

Let the master system
\begin{equation}\label{chenus}
\begin{array}{l}
D_{\ast }^{q_{1}}x_{1}=a\left( x_{2}-x_{1}\right) , \\
D_{\ast }^{q_2} x_{2}=\left( c-a-x_3\right )\widetilde{sgn}(x_{1})+cdx_{2}, \\
D_{\ast }^{q_3} x_{3}=x_{1}\widetilde{sgn}(x_{2})-bx_{3},
\end{array}%
\end{equation}

\noindent and the slave system
\begin{equation}\label{chen_slave}
\begin{array}{l}
D_{\ast }^{q_{1}}y_{1}=a\left( y_{2}-y_{1}\right) +u_1, \\
D_{\ast }^{q_2} y_{2}=\left(c-a-y_3\right)\widetilde{sgn}(y_{1})+cdy_{2}+u_2, \\
D_{\ast }^{q_3} y_{3}=y_{1}\widetilde{sgn}(y_{2})-by_{3}+u_3.%
\end{array}%
\end{equation}

\noindent With the initial conditions $x(0)=[-0.009,-0.012,\allowbreak 0.020]^T$ and $y(0)=x(0)+[0.005,0.005,0.005]^T$\footnote{The relatively small difference between $x(0)$ and $y(0)$ is in agreement with the small size of attractor, which as can be see in, e.g., Fig. \ref{fig10}, is of order of $10^{-2}$, and also avoid long transients before the trajectories reach the attractors.}, if we chose the controller $u=(u_1,u_2,u_3)^T$
\begin{equation*}
\begin{array}{ll}
u_1&=-ae_2+(a-2)e_1,\\
u_2&=(c-a-x_3)\widetilde{sgn}(x_1)-(c-a-y_3)\widetilde{sgn}(y_1)\\
&-(cd+1)e_2,\\
u_3&=x_1\widetilde{sgn}(x_2)-y_1\widetilde{sgn}(y_2)+(b-3)e_3,
\end{array}
\end{equation*}

\noindent where $e_i=y_i-x_i$, $i=1,2,3$, the following error system is obtained
\begin{equation}\label{chen_er}
\begin{array}{l}
D_{\ast }^{q_1} e_{1}=-2e_1, \\
D_{\ast }^{q_2} e_{2}=-e_2,\\
D_{\ast }^{q_3} e_{3}=-3e_3.
\end{array}
\end{equation}

\noindent Here $E$
\[
E=\left(
\begin{array}{ccc}
-2 & 0 & 0 \\
0 & -1 & 0 \\
0 & 0 & -3%
\end{array}%
\right).
\]

\paragraph{Commensurate case}: $q_1=q_2=q_3=0.998$. The eigenvalues of the Jacobian matrix are $\Lambda=(-3,-2,-1)$ with $\alpha_{min}=\pi>1.568=0.998\pi/2$ (Fig. \ref{fig9}a) and the synchronization can be made since the eigenvalues belong to $\Omega$ (Fig. \ref{fig9}b and Fig. \ref{fig9}c).

\paragraph{Incommensurate case}: $q=(0.99,0.98,0.97)$. The fractional error system is asymptotically stable because, the roots of the underlying characteristic equation (\ref{car})
\[
(\lambda^{99}+2)(\lambda^{98}+1)(\lambda^{97}+3)=0,
 \]

\noindent are placed in $\Omega$ (Fig. \ref{fig10}a) and $\alpha_{min}=0.032>0.016=\pi/200$. The two systems synchronize (Fig. \ref{fig10}b and \ref{fig10}c).

\subsection{Shimizu--Morioka system}

The fractional variant of the piece-wise continuous Shimizu–Morioka's system has the following mathematical model \cite{simio,walace}
\begin{equation}\label{sim_master}
\begin{array}{l}
D_{\ast }^{q_{1}}x_{1}=x_{2}, \\
D_{\ast }^{q_2} x_{2}=\widetilde{sgn}(x_1)-x_3\widetilde{sgn}(x_{1})-\alpha x_2 , \\
D_{\ast }^{q_3} x_{3}=x_1^2-\beta x_3,%
\end{array}%
\end{equation}

\noindent where $\alpha=0.75$ and $\beta=0.45$. Beside the origin, the other two equilibrium points are $\tilde{X}^*_{1,2}=(\pm\sqrt{\beta},0,1)$ and the Jacobian is
\[
\tilde{J}=\left(
\begin{array}{ccc}
0 & 1 & 0 \\
0 & -\alpha & -sgn(x_1) \\
2x_1 & 0 & -\beta%
\end{array}%
\right) _{\tilde{X}^{\ast }_{1,2}}.
\]

\noindent Again Properties \ref{prop1} and \ref{prop2} have been used.

\noindent Let (\ref{sim_master}) be the master system and
\begin{equation}\label{sim_slave}
\begin{array}{l}
D_{\ast }^{q_{1}}y_{1}=y_{2}+u_1, \\
D_{\ast }^{q_2} y_{2}=\widetilde{sgn}(y_1)-y_3\widetilde{sgn}(y_{1})-\alpha y_2+u_2 , \\
D_{\ast }^{q_3} y_{3}=y_1^2-\beta y_3+u_3.%
\end{array}%
\end{equation}

\noindent the slave system.

\noindent If we chosen $u_1=-e_1-e_2$, $u_2=-\widetilde{sgn}(y_1)+y_3\widetilde{sgn}(y_1)+\widetilde{sgn}(x_1)-x_3\widetilde{sgn}(x_1)$ and $u_3=-y_1^2+x_1^2$, the error system is
\begin{equation}\label{sim_error}
\begin{array}{l}
D_{\ast }^{q_{1}}e_{1}=-e_1, \\
D_{\ast }^{q_2} e_{2}=-\alpha e_2, \\
D_{\ast }^{q_3} e_{3}=-\beta e_3.%
\end{array}%
\end{equation}

\noindent Let next consider $\tilde{X}_1^*$.

\paragraph{Commensurate case}: $q_1=q_2=q_3=0.9$. Then $\Lambda=(-1.000,-0.750,-0.450)$ and $\alpha_{min}=\pi>1.445=0.9\pi/2$ (Fig. \ref{simi}a). Therefore the systems synchronize (Fig. \ref{simi}b and Fig. \ref{simi}c).
\paragraph{Incommensurate case}: $q=(0.9,1,0.9)$. Then the characteristic polynomial is
\[
\lambda^{28}+9/20\lambda^{19}+3/4\lambda^{18}+27/80\lambda^9+3/5 \sqrt{5},
\]

\noindent and $\Lambda$ containing the 28 roots is included in the stability region $\Omega$ (Fig. \ref{simi_ne}a). $\alpha_{min}=0.314>0.157=\pi/20$ and therefore the systems synchronize (Fig. \ref{simi_ne}b and Fig. \ref{simi_ne}c).

\begin{remark} As can be seen in Fig. \ref{simi} and Fig. \ref{simi_ne}, for this system, the synchronization is of \emph{phase synchronization}-like type (see e.g. \cite{rose}): occurrence of a certain relation between the phases of interacting systems, while the amplitudes
remain chaotic and are, in general, uncorrelated. In Fig. \ref{fig77} the value $+1$ of the cross correlation determined for the components $x_1$ and $y_1$, underlines the perfect phase synchronization. Similar result is obtained for $x_2$ and $y_2$. This phenomenon happens only with respect to $x_{1,2}$ and $y_{1,2}$, while along $x_3$ axis the synchronization is a complete (identical) synchronization.
\end{remark}

\subsection{Chen-Sprott systems}

Finally, let consider the synchronization of two non-identical systems: a Chen system and a Sprott system.

\noindent Let us consider Sprott's system (\ref{spr_apr}) the master system and Chen's system (\ref{chen_slave}) the slave system.

\noindent With $u_1=-a(y_2-y_1)+x_2-3e_1$, $u_2=-(c-a-y_3)\widetilde{sgn}(y_1)-cdy_2+x_3-2e_2$ and $u_3=-y_1\widetilde{sgn}(y_2)+by_3-x_1-x_2-0.5y_3+\widetilde{sgn}(x_1)-e_3$, one obtains the following error system
\begin{equation}
\begin{array}{l}
D_{\ast }^{q_1} e_{1}=-3e_1, \\
D_{\ast }^{q_2} e_{2}=-2e_2,\\
D_{\ast }^{q_3} e_{3}=-e_3,
\end{array}
\end{equation}

\noindent with the error matrix
\[
E=\left(
\begin{array}{ccc}
-3 & 0 & 0 \\
0 & -2 & 0 \\
0 & 0 & -1%
\end{array}%
\right).
\]

\paragraph{Commensurate case}: $q_1=q_2=q_3=0.99$. The eigenvalues are $(-3,-2,-1)$ and $\alpha_{min}=\pi>1.555=0.99\pi/2$ (Fig. \ref{fig11}a). Therefore the error sysem is asymptotically stable and the two systems synchronize (Fig. \ref{fig11}b).

\paragraph {Incommensurate case}: $q=(0.99,0.98,0.97)$. The characteristic equation is
\[
(\lambda^{99}+3)(\lambda^{98}+2)(\lambda^{97}+1)=0,
\]

\noindent for which $\alpha_{min}=0.031>0.015=\pi/200$. Therefore $\Lambda\subset \Omega$ (Fig. \ref{fig11}c) and the systems synchronize, Chen's system following the master system, Sprott's system (Fig. \ref{fig11}d).
\section{Conclusion and discussions}

In this paper we presented a way to synchronize a class of piece-wise systems of fractional-order. This is achievable due to the possibility to approximate continuously the piece-wise functions modeling the underlying systems. The approximation is realized via Cellina's Theorem for set-valued functions, after the piece-wise functions $s$ have been transformed into set-valued functions with the Filippov's regularization. There are several approximation choices. In this paper we use the sigmoid function, which is easy to implement numerically.

Once the systems approximated, the synchronization problem transforms into a synchronization of two (identical or not) continuous systems of fractional-order.

The known stability results for continuous systems of fractional-order, have been adapted for our class of problems.

The sigmoid function used in this paper can be replaced with other continuous approximation.

The approximation algorithm can be used for other purpose too, such as chaos control.

\vspace{3mm}

\appendix
\numberwithin{equation}{section}
\section*{Appendix}
\renewcommand{\theequation}{A.\arabic{equation}}
\setcounter{equation}{0}
\setcounter{theorem}{0}
\renewcommand{\thesubsection}{\Alph{subsection}.}
\subsection{Cellina's Theorem}\label{cela}

\begin{theorem}[Cellina's Theorem \cite{aubin} p. 84 and \cite{aubin2} p. 358]

Let $F:\mathbb{R}^n\rightrightarrows\mathbb{R}^n$ be an u.s.c. function. If the values of $F$ are nonempty and convex, then for every $\varepsilon>0$, there exists a locally Lipschitz single values function $f_\varepsilon:\mathbb{R}^n\rightarrow\mathbb{R}^n$ such that
\begin{equation}\label{cel}
Graph(f_\varepsilon)\subset B(Graph(F),\varepsilon),
\end{equation}
\noindent and for every $x\in \mathbb{R}^n$, $f_\varepsilon(x)$ belongs to teh convex hull of the image of $F$.
\end{theorem}

\noindent In (\ref{cel}), $B$ is the ball in $\mathbb{R}^n$ centered on $F$ and of $\varepsilon$ ray.

\noindent See for example Fig. \ref{fig1}b.

\noindent As known, locally Lipschitz functions are also continuous.
\subsection{Proof of Property 1}\label{ppp}

We shall prove that, for every $\delta>0$, there exists $\varepsilon>0$ such that $|\widetilde{sgn}(x)-1|<\varepsilon$. Solving the inequality, one obtains: $x\in(-\infty,-\delta ln\frac{2-\varepsilon}{\varepsilon})\bigcup(-\delta ln \frac{\varepsilon}{2-\varepsilon},\infty)$. For $\varepsilon\rightarrow0$, $-\delta ln\frac{2-\varepsilon}{\varepsilon}$ and $-\delta ln \frac{\varepsilon}{2-\varepsilon}$ tend to 0. Therefore $\widetilde{sgn}(x)\approx Sgn(x)$ for $x\not\in\mathcal{V}=(-\delta ln\frac{2-\varepsilon}{\varepsilon},-\delta ln\frac{\varepsilon}{2-\varepsilon})$ and, consequently, $ \tilde{X}^*\approx X^*$.
\qed

\vspace{3mm}

\noindent For $\varepsilon\in\{0,2\}$, $ln$ is not defined. However, our interest concerns $\varepsilon\neq 2,0$.
\subsection{Proof of Property 2}\label{pppp}

We prove that for every $\delta>0$ there exists $\varepsilon>0$ such that $|\frac{d}{dx}\widetilde{sgn}(x)|<\varepsilon$. By solving the inequality $\frac{d}{dx}\widetilde{sgn}(x)=\frac{2e^{x/d}}{\delta(e^{x/d}+1)^2}<\varepsilon$ (Fig. \ref{fig55}a), one obtains $|x|>\delta ln\frac{\sqrt{1-2\delta\varepsilon}+\delta\varepsilon-1}{\delta\varepsilon}$ (Fig. \ref{fig55}b). On the other side, for $x\neq0$, $\frac{d}{dx}sgn(x)=0$. Therefore, for $x\not\in\mathcal{V}=(-\delta ln\frac{\sqrt{1-2\delta\varepsilon}+\delta\varepsilon-1}{\delta\varepsilon}, \delta ln\frac{\sqrt{1-2\delta\varepsilon}+\delta\varepsilon-1}{\delta\varepsilon})$, $\frac{d}{dx}\widetilde{sgn}(x)\approx \frac{d}{dx}sgn(x)$ and therefore $\tilde{J}\approx J$ for $x\not\in\mathcal{V}$.
\qed
\vspace{3mm}

Above, beside Assumption (\textbf{H2}), it is supposed that $x_isgn(x_j)$, for $i,j\in{1,2,...,n}$, is differentiable on the interior of $\mathcal{D}_i$.

\subsection{Hausdorff distance}\label{haha}
The Hausdorff distance (or Hausdorff metric) $D_H$ measures how far two compact nonempty subsets of the considered
metric space $\mathbb{R}^n$ are from each other \cite{falco}. The Hausdorff distance between two curves in (here $\mathbb{R}^n$) is defined as the maximum distance to the closest point between the curves. If the curves are defined, as in our case, as the sets of ordered pair of coordinates
$A=\{P_1,P_2,...,P_{m_1}\}$ and $B=\{Q_1,Q_2,...,Q_{m_2}\}$ with $P_i=(x_1,x_2,...,x_n)$ and $Q_i=\{y_1,y_2,...,y_n)$, then $D_H$ is expressed as follows
\[
D_H(P,Q)=max\{d(P,Q),d(Q,P)\},
\]

\noindent where $d(P,Q)$ (generally differen to $d(Q,P)$) has the expression
\[
d(P,Q)=\underset{i}{max}\{d(P_i,Q)\},
\]

\noindent and
\[
d(P_i,Q)=\underset{j}{max}||P_i-Q_j||.
\]

%

\newpage

\newpage

\begin{figure*}
\begin{center}
  \includegraphics[clip,width=0.7\textwidth] {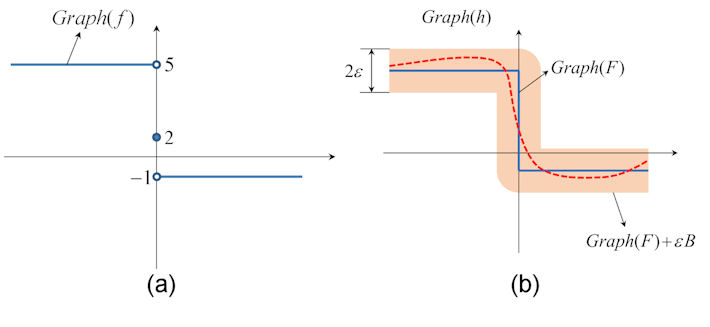}
\caption{a) Graph of $f(x)=2-3sgn(x)$. b) Graph of set-valued function $F(x)=2-3Sgn(x)$ (blue) and continuous approximation of $F(x)$ (red). }
\label{fig1}
\end{center}
\end{figure*}

\begin{figure*}
\begin{center}
  \includegraphics[clip,width=0.7\textwidth] {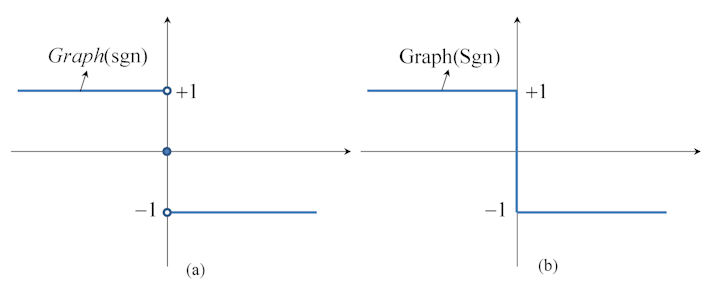}
\caption{a) Graph of $sgn$ function. b) Graph of $Sgn$ function.}
\label{fig2}
\end{center}
\end{figure*}

\begin{figure*}
\begin{center}
  \includegraphics[clip,width=0.8\textwidth] {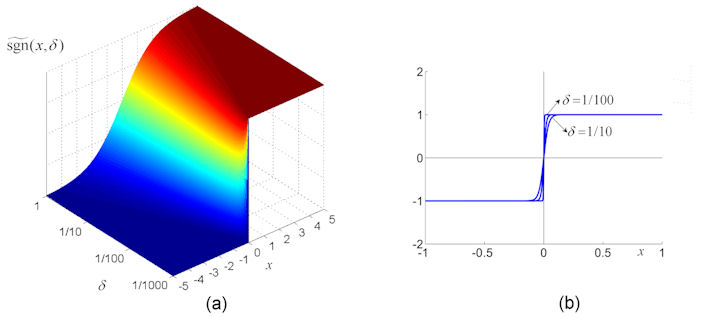}
\caption{Graph of sigmoid function $\widetilde{sgn}$. a) Dependence on $\delta$. b) Graph of $\widetilde{sgn}$ for several values $\delta$.}
\label{fig3}
\end{center}
\end{figure*}

\begin{figure*}
\begin{center}
  \includegraphics[clip,width=0.4\textwidth] {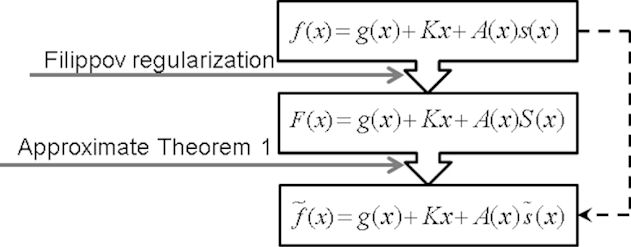}
\caption{Continuous approximation algorithm. Dotted line shows the direct way to approximate the discontinuous components $s$. }
\label{schema}
\end{center}
\end{figure*}

\begin{figure*}
\begin{center}
  \includegraphics[clip,width=1\textwidth] {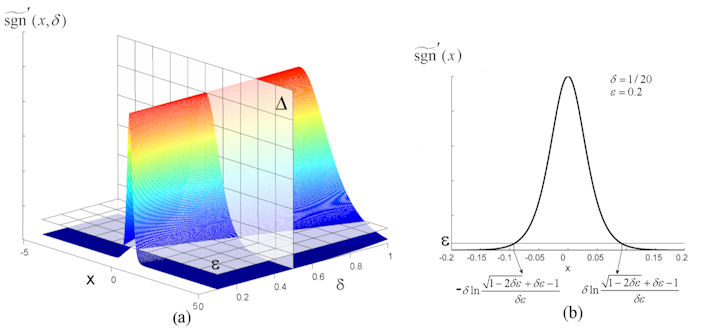}
\caption{a) Graph of $\frac{d}{dx}\widetilde{sgn}$, function on $\delta$. $\Delta$ and $\varepsilon$ represent vertical and horizontal planes through some $\delta$ and $\varepsilon$ values. b) Graph of $\frac{d}{dx}\widetilde{sgn}$ for $\delta=1/10$ and $\varepsilon=0.2$ ($\delta$ and $\varepsilon$ have been chosen larger for a clearer image).}
\label{fig55}
\end{center}
\end{figure*}

\begin{figure*}
\begin{center}
  \includegraphics[clip,width=0.8\textwidth] {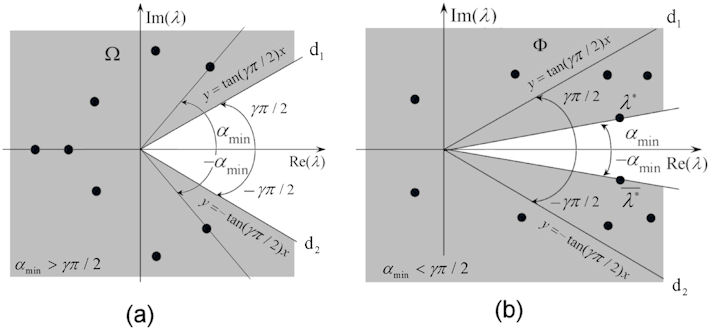}
\caption{a) Stability domain $\Omega$ determined by $\alpha_{min}>\gamma\pi/2$. b) Instability domain $\Phi$ determined by $\alpha_{min}\leq\gamma\pi/2$. $d_{1,2}$ are the stability separatrices.}
\label{fig66}
\end{center}
\end{figure*}

\begin{figure*}
\begin{center}
  \includegraphics[clip,width=0.9\textwidth] {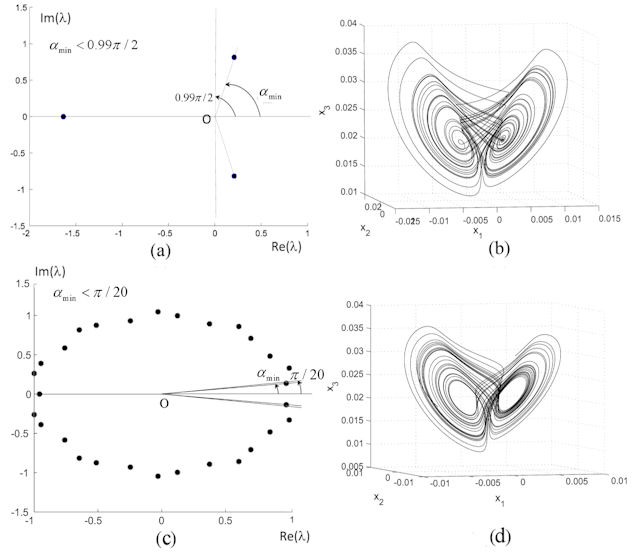}
\caption{Chaotic Piece-wise linear Chen system for: top $q=0.99$; bottom $q=(1,0.9,1)$. a) Commensurate case $q=0.99$, $\alpha_{min}<0.99\pi/2$ and the system is unstable. b) The underlying chaotic attractor. c) Incommensurate $q=(1,0.9,1)$, $\alpha_{min}<\pi/20$ and the system is unstable. d) The underlying chaotic attractor. }
\label{fig6}
\end{center}
\end{figure*}

\begin{figure*}
\begin{center}
  \includegraphics[clip,width=0.8\textwidth] {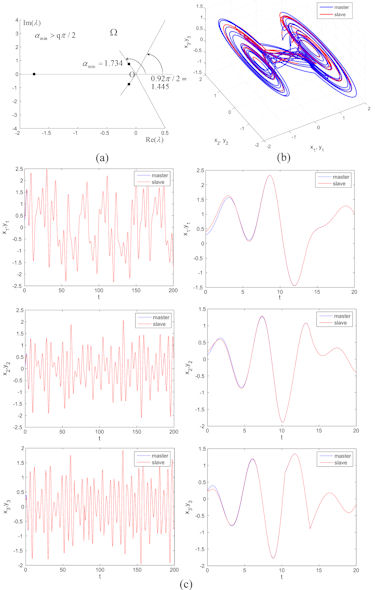}
\caption{Synchronization of Sprott system for $q=0.92$. a) $\alpha_{min}>0.92\pi/2$ and the error system is stable and synchronization holds. b) Phase plot of both trajectories (in blue the master system, in red the slave system). c) Time series: left column $t\in[0,200]$, right column $t\in[0,20]$.}
\label{fig7}
\end{center}
\end{figure*}

\begin{figure*}
\begin{center}
  \includegraphics[clip,width=0.8\textwidth] {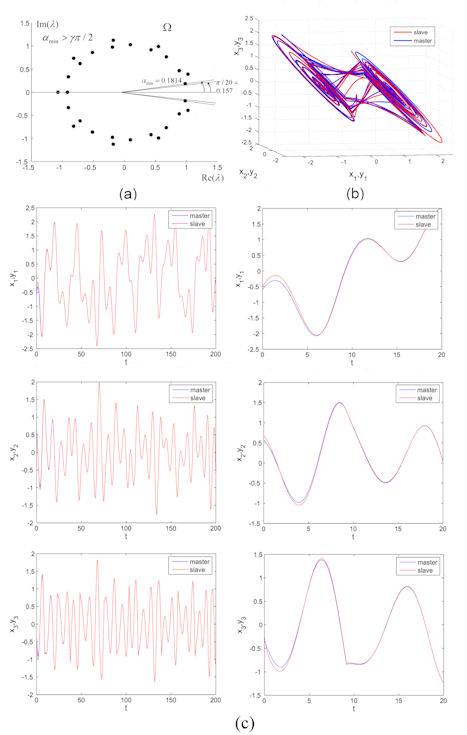}
\caption{Synchronization of Sprott system for $q=(0.9,1,0.9)$. a) $\alpha_{min}>\pi/20$ and the error system is stable and synchronization holds. b) Phase plot of both trajectories (in blue the master system, in red the slave system). c) Time series: left column $t\in[0,200]$, right column $t\in[0,20]$.}
\label{fig8}
\end{center}
\end{figure*}

\begin{figure*}
\begin{center}
  \includegraphics[clip,width=0.8\textwidth] {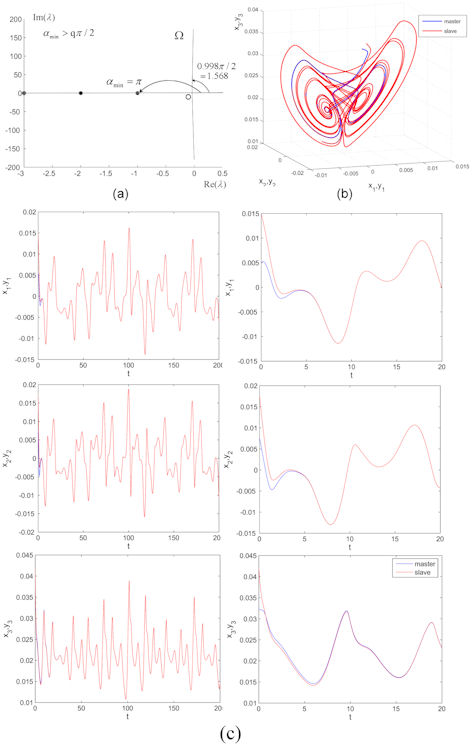}
\caption{Synchronization of piece-wise Chen system for $q_1=q_2=q_3=0.998$. a) $\alpha_{min}>0.998\pi/2$ and the error system is stable and synchronization holds. b) Phase plot of both trajectories (in blue the master system, in red the slave system). c) Time series: left column $t\in[0,200]$, right column $t\in[0,20]$.}
\label{fig9}
\end{center}
\end{figure*}

\begin{figure*}
\begin{center}
  \includegraphics[clip,width=0.8\textwidth] {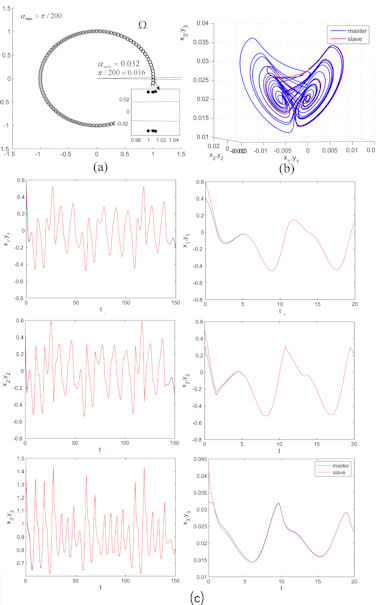}
\caption{Synchronization of piece-wise Chen system for $q=(0.99,0.98,0.97)$. a) $\alpha_{min}>\pi/200$ and the error system is stable and synchronization holds. Detail reveal the positions of the roots versus the separatrices. b)Phase plot of both trajectories (in blue the master system, in red the slave system). c) Time series: left column $t\in[0,200]$, right column $t\in[0,20]$.}
\label{fig10}
\end{center}
\end{figure*}

\begin{figure*}
\begin{center}
  \includegraphics[clip,width=0.8\textwidth] {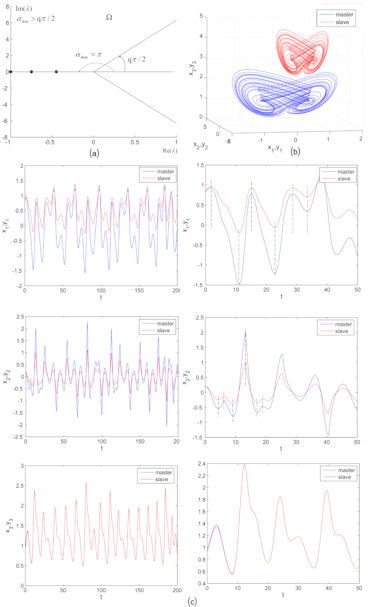}
\caption{Synchronization of piece-wise Shimizu-–Morioka system for $q_1=q_2=q_3=0.9$. a) $\alpha_{min}=\pi>0.9\pi/2$ and the error system is stable and synchronization holds. b)Phase plot of both trajectories (in blue the master system, in red the slave system). c) Time series: left column $t\in[0,200]$, right column $t\in[0,20]$. Dotted lines underline phase synchronization.}
\label{simi}
\end{center}
\end{figure*}

\begin{figure*}
\begin{center}
  \includegraphics[clip,width=0.8\textwidth] {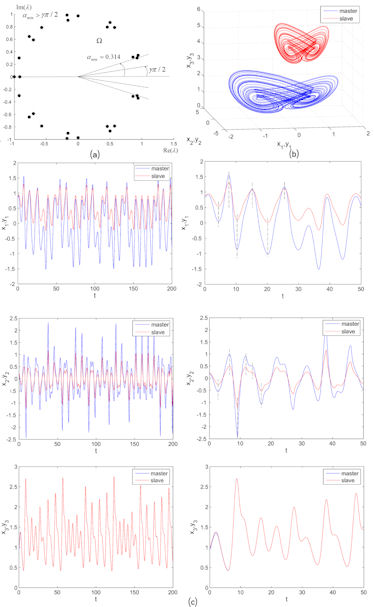}
\caption{Synchronization of piece-wise Shimizu–-Morioka system for $q=(0.9,1,0.9)$. a) $\alpha_{min}=0.314>\pi/20$ and the error system is stable and synchronization holds. b) Plot of both trajectories translated along $x_3$ axis (in blue the master system, in red the slave system). c) Time series: left column $t\in[0,200]$, right column $t\in[0,20]$. Dotted lines underline phase synchronization.}
\label{simi_ne}
\end{center}
\end{figure*}

\begin{figure*}[h]
\begin{center}
  \includegraphics[clip,width=0.5\textwidth] {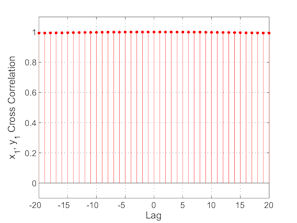}
\caption{Cross correlation for the components $x_1$ and $y_1$ of two Shimizu–-Morioka synchronized systems. $x_1$ and $y_1$ are in phase since the value of the cross correlation is $+1$.}
\label{fig77}
\end{center}
\end{figure*}

\begin{figure*}
\begin{center}
  \includegraphics[clip,width=0.8\textwidth] {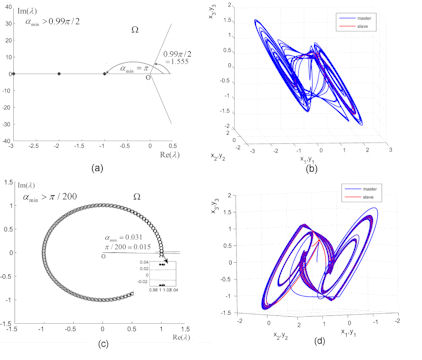}
\caption{Synchronization of Sprott's system and piece-wise Chen's system. a) For $q_1=q_2=q_3=0.99$, $\alpha_{min}>0.99\pi/2$ and the error system is stable and synchronization holds. b) Phase plot of both trajectories translated along $x_3$ axis (in blue the master system, in red the slave system). c) For $q=(0.99,098,0.97)$, $\alpha_{min}>\pi/200$ and the systems synchronize. Detail reveals the positions of the roots versus the separatrices. d) Phase plot of both trajectories (in blue the master system, in red the slave system).}
\label{fig11}
\end{center}
\end{figure*}

\end{document}